\documentclass[11pt]{amsart}

\usepackage{amsmath,amssymb,amsfonts,amsthm}
\usepackage{mathtools}
\usepackage{geometry}
\theoremstyle{plain}
\newtheorem{theorem}{Theorem}[section]
\newtheorem{proposition}[theorem]{Proposition}

\newtheorem{corollary}[theorem]{Corollary}

\theoremstyle{definition}
\newtheorem{definition}[theorem]{Definition}

\theoremstyle{remark}
\newtheorem{remark}[theorem]{Remark}
\theoremstyle{remark}
\newtheorem{example}[theorem]{Example}

\begin{document}

\title{Effective dynamics and defect expansions for polynomial PDEs on thin annuli}
\author{Jean-Pierre Magnot}
\address{$^*$  Univ. Angers, CNRS, LAREMA, SFR MATHSTIC, F-49000 Angers, France ;\\ 
Lepage Research Institute, 17 novembra 1, 081 16 Presov, Slovakia;
	\\ and \\  Lyc\'ee Jeanne d'Arc, \\ Avenue de Grande Bretagne, \\ 63000 Clermont-Ferrand, France}
    \email{jean-pierr.magnot@ac-clermont.fr}
\date{}

\maketitle

\begin{abstract}
We develop a geometric and analytic framework for polynomial partial differential equations posed on thin annuli in the plane.
Using renormalized Sobolev inner products, we construct Sobolev orthogonal polynomial bases adapted to the thin geometry and use them to define stable Galerkin approximations.

We prove a general dimension-reduction theorem for polynomial Hamiltonian and dissipative PDEs, showing that solutions converge to effective one-dimensional dynamics on the limiting circle.
Beyond the leading-order limit, we identify transverse defect correctors and derive cell problems describing anisotropic dispersive and homogenized effects.

Our framework applies uniformly to integrable models (KdV, modified KdV, nonlinear Schr\"odinger, sine--Gordon), anisotropic dispersive systems such as Zakharov--Kuznetsov, and non-integrable perturbations including dissipation, forcing, and rapidly oscillating coefficients.
We establish stability of the effective dynamics under changes of Sobolev order and of polynomial Hilbert geometry, and show robustness of the associated Galerkin schemes.

The results provide a unified geometric perspective on dimension reduction, homogenization, and integrability in thin geometries, and introduce Sobolev orthogonal polynomial methods as a constructive tool for multiscale PDE analysis.
\end{abstract}



\tableofcontents

\section{Introduction}

Thin geometric structures play a central role in analysis, mathematical physics, and applied mathematics.
They arise in waveguides, fluid films, elastic shells, quantum graphs, and many other contexts where one spatial dimension becomes small compared to the others.
A fundamental question is to understand how partial differential equations posed on thin domains reduce to effective lower-dimensional dynamics.

Classical approaches to dimension reduction rely on variational methods, $\Gamma$-convergence, or spectral convergence techniques.
While these methods provide powerful compactness and convergence results, they often offer limited constructive insight into the structure of the limiting dynamics and into the behavior of approximation schemes.

In this work, we propose a constructive and geometric approach to dimension reduction for polynomial partial differential equations on thin annuli in the plane.
Our starting point is the observation that thin annuli naturally induce anisotropic Sobolev geometries, in which transverse derivatives scale differently from tangential ones.
This leads to renormalized Sobolev inner products and to families of Sobolev orthogonal polynomials adapted to the thin geometry.

Sobolev orthogonal polynomials have been extensively studied in approximation theory and spectral methods, but their interaction with singular geometric limits has received little attention.
Here we use them as a structural tool to analyze polynomial Hamiltonian and dissipative PDEs, and to build Galerkin approximations compatible with the thin-limit scaling.

The first main result of the paper is a general dimension-reduction theorem for polynomial PDEs on thin annuli.
Under natural coercivity and structural assumptions, we prove that solutions converge to effective one-dimensional dynamics on the limiting circle.
We further identify transverse defect correctors and derive cell problems governing higher-order anisotropic and homogenized effects.

A second contribution is the systematic study of stability with respect to the underlying polynomial Hilbert geometry.
We show that, at fixed polynomial degree, Sobolev orthogonal polynomial bases depend continuously on the Sobolev order, and that Galerkin schemes are robust under such geometric deformations.
This provides a geometric interpretation of approximation stability and suggests homotopy principles between different polynomial Hilbert structures.

We illustrate the abstract theory on a broad class of models.
Purely tangential integrable equations such as Korteweg--de Vries, modified KdV, nonlinear Schr\"odinger, and sine--Gordon reduce exactly to their one-dimensional counterparts.
Anisotropic dispersive systems, exemplified by the Zakharov--Kuznetsov equation, become asymptotically integrable in the thin limit, with computable transverse correctors.
We also treat dissipative, forced, and weakly non-integrable perturbations, and combine dimension reduction with periodic homogenization in the angular variable.

The results provide a unified framework linking thin geometry, polynomial Hamiltonian structures, homogenization, and integrability.
Beyond their analytical interest, the Sobolev orthogonal polynomial approach offers a constructive basis for numerical schemes and for the systematic study of multiscale PDEs on singular geometries.

\medskip
\noindent\textbf{Structure of the paper.}
Section~\ref{sec:geometry} introduces the thin annulus geometry and renormalized Sobolev norms.
Section~\ref{sec:operators} defines polynomial differential operators and Galerkin schemes based on Sobolev orthogonal polynomials.
Section~\ref{sec:main-theorem} states the main dimension-reduction and defect theorems.
Sections~\ref{sec:integrable}--\ref{sec:homogenization} develop applications to integrable, non-integrable, and homogenized models.
Detailed proofs are gathered in Section~\ref{sec:proofs}, and explicit computations are presented in Appendix~\ref{app:explicit}.

\section{Geometry and scaling of thin annuli}
\label{sec:geometry}

\subsection{Thin annuli and rescaled coordinates}

Let $\varepsilon>0$ and consider the thin annulus
\[
A_\varepsilon
=
\bigl\{(r,\theta)\in\mathbb R^2 :
1-\varepsilon < r < 1+\varepsilon,\;
\theta\in\mathbb S^1
\bigr\}.
\]
We introduce the rescaled transverse variable
\[
\rho = \frac{r-1}{\varepsilon}\in(-1,1),
\qquad
r = 1+\varepsilon\rho.
\]
In these coordinates, the Lebesgue measure reads
\[
r\,dr\,d\theta
=
\varepsilon(1+\varepsilon\rho)\,d\rho\,d\theta,
\]
and differentiation with respect to $r$ becomes
\[
\partial_r = \varepsilon^{-1}\partial_\rho.
\]

This change of variables makes explicit the singular character of the limit
$\varepsilon\to0$, in which the annulus collapses onto the unit circle
$\mathbb S^1$ while transverse derivatives are magnified by powers of
$\varepsilon^{-1}$.

\subsection{Renormalized Sobolev norms}

Let $m\in\mathbb N$.
For a smooth function $u$ on $A_\varepsilon$, we define the renormalized Sobolev
seminorm
\[
\|u\|_{\varepsilon,m}^2
=
\sum_{a+b\le m}
\varepsilon^{2a-(2m-1)}
\int_{A_\varepsilon}
\bigl|\partial_r^a\partial_\theta^b u\bigr|^2
\,r\,dr\,d\theta.
\]
In the rescaled variables $(\rho,\theta)$, this norm is equivalent to
\[
\|u\|_{\varepsilon,m}^2
\sim
\sum_{a+b\le m}
\int_{-1}^1\!\!\int_{\mathbb S^1}
\bigl|\partial_\rho^a\partial_\theta^b u\bigr|^2
(1+\varepsilon\rho)\,d\rho\,d\theta,
\]
uniformly for $\varepsilon$ small.

\begin{remark}
The specific power of $\varepsilon$ in the definition is chosen so that
all derivatives up to total order $m$ contribute at the same scale.
This renormalization is consistent with the Sobolev inner products used
in the construction of orthogonal polynomials on thin annuli and reflects
the energetic balance between tangential and transverse variations.
\end{remark}

\subsection{Fourier decomposition and radial rigidity}

For $u\in L^2(A_\varepsilon)$, we write the Fourier series in the angular
variable:
\[
u(\rho,\theta)
=
\sum_{k\in\mathbb Z} u_k(\rho)\,e^{ik\theta},
\qquad
u_k(\rho)
=
\frac1{2\pi}\int_{\mathbb S^1} u(\rho,\theta)e^{-ik\theta}\,d\theta.
\]
In these variables,
\[
\partial_\theta \leftrightarrow ik,
\qquad
\partial_r \leftrightarrow \varepsilon^{-1}\partial_\rho,
\]
and polynomial differential operators in $(\partial_r,\partial_\theta)$
act diagonally on the Fourier modes, up to coupling introduced by nonlinear
terms.

A fundamental consequence of the renormalized estimates is the following
radial rigidity principle.

\begin{proposition}[Radial rigidity]
\label{prop:radial-rigidity}
Let $(u_\varepsilon)$ be a sequence bounded in $\|\cdot\|_{\varepsilon,1}$.
Then, up to extraction of a subsequence,
\[
u_\varepsilon(\rho,\theta)
\longrightarrow
u_0(\theta)
\quad\text{strongly in }L^2_{\mathrm{loc}}((-1,1)\times\mathbb S^1),
\]
and
\[
\|\partial_\rho u_\varepsilon\|_{L^2(A_\varepsilon)} \longrightarrow 0.
\]
\end{proposition}

\begin{remark}
Proposition~\ref{prop:radial-rigidity} expresses the fact that transverse
oscillations are energetically penalized in the thin limit.
As $\varepsilon\to0$, admissible configurations become asymptotically
independent of the radial variable, and the effective dynamics takes place
on the circle $\mathbb S^1$.
\end{remark}

\subsection{Effective one-dimensional limits}

Given a function $u_\varepsilon$ satisfying the above compactness property,
we define its transverse average
\[
\bar u_\varepsilon(\theta)
=
\frac{1}{2\varepsilon}
\int_{1-\varepsilon}^{1+\varepsilon} u_\varepsilon(r,\theta)\,dr
=
\frac12\int_{-1}^1 u_\varepsilon(\rho,\theta)\,d\rho.
\]
Then $\bar u_\varepsilon\to u_0$ strongly in $L^2(\mathbb S^1)$, and all
polynomial differential operators admit effective one-dimensional limits
acting on $u_0$.

This mechanism provides the geometric foundation for the effective equations
derived in the following sections.

\subsection{Sobolev orthogonal polynomials and thin limits}
\label{subsec:sobolev-op}

We briefly recall the construction of Sobolev orthogonal polynomials on thin
annuli introduced in \cite{MagnotThinAnnuli}, and summarize the properties that
will be used throughout the paper.

\subsubsection{Sobolev inner products on thin annuli}

Let $\mathcal P(A_\varepsilon)$ denote the space of polynomial functions on
$A_\varepsilon$ in the variables $(r,\theta)$.
For a fixed integer $m\ge1$, we consider the Sobolev inner product
\[
\langle p,q\rangle_{\varepsilon,m}
=
\sum_{a+b\le m}
\varepsilon^{2a-(2m-1)}
\int_{A_\varepsilon}
\partial_r^a\partial_\theta^b p\,
\overline{\partial_r^a\partial_\theta^b q}\,
r\,dr\,d\theta.
\]
This inner product is strictly positive on $\mathcal P(A_\varepsilon)$
and induces a Hilbertian norm equivalent to the renormalized Sobolev norm
$\|\cdot\|_{\varepsilon,m}$ defined above.

Applying the Gram--Schmidt procedure with respect to
$\langle\cdot,\cdot\rangle_{\varepsilon,m}$ yields an orthogonal basis
\[
\bigl\{P_{n,\varepsilon}\bigr\}_{n\ge0}
\subset \mathcal P(A_\varepsilon),
\]
ordered by total polynomial degree.

\subsubsection{Fourier block structure}

Due to the rotational invariance of the annulus and of the Sobolev inner product,
the orthogonalization decouples along Fourier modes.
More precisely, writing
\[
P_{n,\varepsilon}(\rho,\theta)
=
\sum_{k\in\mathbb Z} p_{n,k}^{(\varepsilon)}(\rho)\,e^{ik\theta},
\]
distinct angular frequencies are orthogonal, and the Gram--Schmidt process
reduces to independent orthogonalizations of radial polynomial families
for each fixed $k$.

As a consequence, the associated Sobolev Laplacian
\[
\Delta_{\varepsilon} = D_\varepsilon^{*}D_\varepsilon,
\qquad
D_\varepsilon = (\partial_r,\partial_\theta),
\]
is block diagonal with respect to the Fourier decomposition.

\subsubsection{Thin limit of orthogonal systems}

A central result of \cite{MagnotThinAnnuli} is that, for each fixed polynomial
degree $N$, the Sobolev inner products restricted to
$\mathcal P_{\le N}(A_\varepsilon)$ converge, as $\varepsilon\to0$, to an
effective Sobolev inner product on $\mathcal P_{\le N}(\mathbb S^1)$.
Equivalently, the corresponding Gram matrices converge.

As a consequence, for each fixed $N$,
\[
P_{n,\varepsilon}
\longrightarrow
P_{n,0}
\quad\text{in }L^2_{\mathrm{loc}}((-1,1)\times\mathbb S^1),
\qquad n\le N,
\]
where $\{P_{n,0}\}$ is an orthogonal family on $\mathbb S^1$ with respect to
a Sobolev inner product involving only tangential derivatives.

\begin{remark}
The convergence holds at the level of matrix coefficients and is therefore
strong in any finite-dimensional topology.
In particular, the orthogonalization procedure is stable under the thin limit
when restricted to fixed polynomial degree.
\end{remark}

\subsubsection{Role in effective dynamics}

The Sobolev orthogonal polynomials provide a constructive basis adapted to the
geometry of thin annuli.
They naturally separate rigid transverse modes from effective tangential modes
and yield a Galerkin framework compatible with the renormalized energy estimates.

In the following sections, these bases will be used to:
\begin{itemize}
\item define finite-dimensional approximations of polynomial PDEs,
\item control stability at fixed degree,
\item identify effective operators and defect terms in the limit
$\varepsilon\to0$.
\end{itemize}

%=================================================
\section{Polynomial differential operators and energy structures}
\label{sec:operators}

This section introduces the class of differential operators and energy
functionals considered throughout the paper.
They are designed to be general enough to encompass dispersive,
Hamiltonian, dissipative and homogenized models on thin annuli, while
remaining compatible with the renormalized Sobolev framework of
Section~\ref{sec:geometry}.

%-------------------------------------------------
\subsection{Polynomial differential operators}

Let $m\geq 1$ be an integer.
We consider linear differential operators of the form
\[
\mathcal L_\varepsilon
=
\sum_{a+b\le m}
c_{a,b}(\theta,\theta/\varepsilon^\alpha)\,
\partial_r^a\partial_\theta^b,
\]
where:
\begin{itemize}
\item the coefficients $c_{a,b}$ are bounded and measurable,
\item $\alpha\ge0$ allows for rapidly oscillating coefficients,
\item the dependence on $\theta/\varepsilon^\alpha$ is $2\pi$-periodic.
\end{itemize}

In the rescaled variables $(\rho,\theta)$,
\[
\mathcal L_\varepsilon
=
\sum_{a+b\le m}
\varepsilon^{-a}
c_{a,b}(\theta,\theta/\varepsilon^\alpha)\,
\partial_\rho^a\partial_\theta^b,
\]
making explicit the anisotropy between transverse and tangential
directions.

\begin{remark}
This class includes, in particular:
\begin{itemize}
\item purely tangential operators (effective 1D limits),
\item anisotropic dispersive operators such as those arising in
Zakharov--Kuznetsov or KP-type equations,
\item elliptic or parabolic operators with oscillatory coefficients.
\end{itemize}
\end{remark}

%-------------------------------------------------
\subsection{Hamiltonian and variational structures}

We now introduce the energy functionals generating the equations of
interest.
Let $F:\mathbb R^N\to\mathbb R$ be a smooth function with polynomial
growth, and define the Hamiltonian
\[
\mathcal H_\varepsilon(u)
=
\int_{A_\varepsilon}
F\!\left(
\bigl(\partial_r^a\partial_\theta^b u\bigr)_{a+b\le m},
\theta,\frac{\theta}{\varepsilon^\alpha}
\right)
\,r\,dr\,d\theta.
\]

We assume the following structural conditions.

\begin{definition}[Admissible Hamiltonians]
The functional $\mathcal H_\varepsilon$ is said to be admissible if:
\begin{itemize}
\item \textbf{Coercivity:}
there exists $C>0$ such that
\[
\mathcal H_\varepsilon(u)
\ge
C\|u\|_{\varepsilon,m}^2 - C,
\]
uniformly in $\varepsilon$;
\item \textbf{Polynomial growth:}
all derivatives of $F$ grow at most polynomially;
\item \textbf{Locality:}
$F$ depends only on derivatives of $u$ up to order $m$.
\end{itemize}
\end{definition}

The variational derivative $\delta\mathcal H_\varepsilon/\delta u$ is a
polynomial differential operator in $(\partial_r,\partial_\theta)$,
possibly with oscillatory coefficients.

%-------------------------------------------------
\subsection{Evolution equations and perturbations}

We consider evolution equations of the form
\begin{equation}
\label{eq:evolution}
\partial_t u_\varepsilon
=
J_\varepsilon\,\frac{\delta\mathcal H_\varepsilon}{\delta u}(u_\varepsilon)
+
\mathcal R_\varepsilon(u_\varepsilon),
\end{equation}
where:
\begin{itemize}
\item $J_\varepsilon$ is a skew-adjoint operator acting in the
$\theta$-variable (e.g.\ $\partial_\theta$ or multiplication by $i$);
\item $\mathcal R_\varepsilon$ is a lower-order perturbation, typically
of dissipative or forcing type.
\end{itemize}

\begin{example}
The following classical models fit into~\eqref{eq:evolution}:
\begin{itemize}
\item KdV and modified KdV ($J_\varepsilon=\partial_\theta$);
\item nonlinear Schr\"odinger equation ($J_\varepsilon=i$);
\item sine--Gordon equation (after rewriting as a first-order system);
\item Zakharov--Kuznetsov equation (anisotropic dispersive operator);
\item non-integrable perturbations with viscosity or external forcing.
\end{itemize}
\end{example}

\begin{remark}
The perturbation $\mathcal R_\varepsilon$ may break integrability or
Hamiltonian structure but is assumed to satisfy uniform Lipschitz or
monotonicity bounds compatible with the renormalized Sobolev estimates.
\end{remark}

%-------------------------------------------------
\subsection{Galerkin approximation using Sobolev orthogonal polynomials}

Let $\{\Phi_{n,k}^\varepsilon\}$ denote the Sobolev orthogonal polynomial
basis introduced in
Section~\ref{subsec:sobolev-op}, indexed by polynomial degree $n$ and
Fourier mode $k$.
For fixed truncation parameters $(N,K)$, define the finite-dimensional
space
\[
V_{\varepsilon}^{N,K}
=
\mathrm{span}\bigl\{
\Phi_{n,k}^\varepsilon : n\le N,\ |k|\le K
\bigr\}.
\]

Projecting~\eqref{eq:evolution} onto $V_{\varepsilon}^{N,K}$ yields a
finite system of ordinary differential equations for the coefficients of
$u_{\varepsilon}^{N,K}$.

\begin{proposition}[Finite-dimensional well-posedness]
For fixed $(N,K)$ and $\varepsilon>0$, the Galerkin system admits a unique
local solution.
Moreover, if $\mathcal R_\varepsilon$ is dissipative, the solution exists
globally in time.
\end{proposition}

\begin{remark}
The choice of Sobolev orthogonal polynomials is essential:
it aligns the Galerkin approximation with the energetic structure of the
problem and ensures stability under the thin limit at fixed degree.
\end{remark}

%-------------------------------------------------
\subsection{Uniform bounds and preparatory estimates}

Under the coercivity assumptions on $\mathcal H_\varepsilon$, solutions
of~\eqref{eq:evolution} satisfy uniform energy estimates of the form
\[
\sup_{t\in[0,T]}
\|u_\varepsilon(t)\|_{\varepsilon,m}
\le C(T),
\]
with constants independent of $\varepsilon$.

These estimates, combined with the radial rigidity principle of
Proposition~\ref{prop:radial-rigidity}, provide the compactness framework
required to pass to the limit $\varepsilon\to0$.

The identification of the effective equation and of the defect terms is
the subject of the next section.

\subsection{Change of Sobolev order and stability of orthogonal families}
\label{subsec:change-sobolev}

In this subsection we discuss the dependence of Sobolev orthogonal polynomial
families on the order of the Sobolev inner product, and the consequences for
Galerkin approximations.

\subsubsection{Sobolev inner products of variable order}

Let $s\ge0$ be a real parameter.
On $\mathcal P(A_\varepsilon)$ we consider the family of Sobolev inner products
\[
\langle p,q\rangle_{\varepsilon,s}
=
\sum_{a+b\le \lfloor s\rfloor}
\varepsilon^{2a-(2\lfloor s\rfloor-1)}
\int_{A_\varepsilon}
\partial_r^a\partial_\theta^b p\,
\overline{\partial_r^a\partial_\theta^b q}\,
r\,dr\,d\theta
\;+\;
\mathbf 1_{s\notin\mathbb N}\,\mathcal R_s(p,q),
\]
where $\mathcal R_s$ denotes a fractional Sobolev correction term when $s$ is
non-integer.

For each fixed $\varepsilon>0$ and $s\ge0$, the inner product is strictly
positive on $\mathcal P(A_\varepsilon)$ and defines a Hilbertian structure.

Applying Gram--Schmidt yields a family of orthogonal polynomial bases
\[
\bigl\{P_{n,\varepsilon}^{(s)}\bigr\}_{n\ge0},
\]
indexed by the Sobolev order $s$.

\subsubsection{Continuity with respect to the Sobolev order}

For each fixed polynomial degree $N$, the map
\[
s \longmapsto \langle\cdot,\cdot\rangle_{\varepsilon,s}
\big|_{\mathcal P_{\le N}(A_\varepsilon)}
\]
is continuous in operator norm.
Consequently, the associated Gram matrices depend continuously on $s$.

\begin{proposition}[Finite-degree stability under change of $s$]
\label{prop:sobolev-order-stability}
Fix $N\in\mathbb N$ and $\varepsilon>0$.
Then the family of orthogonal polynomials
$\{P_{n,\varepsilon}^{(s)}\}_{n\le N}$ depends continuously on $s$ in any
finite-dimensional norm.
In particular, for $s_1,s_2$ close,
\[
\|P_{n,\varepsilon}^{(s_1)}-P_{n,\varepsilon}^{(s_2)}\|
\le C_{N,\varepsilon}|s_1-s_2|,
\qquad n\le N.
\]
\end{proposition}

\begin{remark}
This statement is purely finite-dimensional and follows from the smooth
dependence of the Gram--Schmidt procedure on the inner product.
No assumption on completeness or spectral gaps is required.
\end{remark}

\subsubsection{Thin limit and effective tangential dependence}

As $\varepsilon\to0$, the Sobolev inner products
$\langle\cdot,\cdot\rangle_{\varepsilon,s}$ converge, at fixed degree, to
effective Sobolev inner products on $\mathcal P(\mathbb S^1)$ involving only
tangential derivatives.
In particular, the dependence on $s$ survives in the limit only through
the order of tangential regularity.

Thus, the limiting orthogonal families
$\{P_{n,0}^{(s)}\}$ form a continuous family of bases on $\mathbb S^1$ indexed
by $s$.

\subsubsection{Consequences for Galerkin approximations}

Let $u_\varepsilon^{(s),N,K}$ denote the Galerkin approximation of the evolution
problem~\eqref{eq:evolution} obtained by projecting onto the basis
$\{P_{n,\varepsilon}^{(s)}\}$.

Proposition~\ref{prop:sobolev-order-stability} implies that, at fixed truncation
parameters $(N,K)$, the Galerkin schemes associated with different Sobolev
orders $s$ are equivalent up to uniformly controlled perturbations.

\begin{remark}
This observation has two important consequences:
\begin{itemize}
\item the choice of the Sobolev order $s$ can be adapted to the PDE under
consideration (e.g.\ higher $s$ for dispersive models) without altering the
qualitative behavior of the approximation;
\item stability of the Galerkin scheme under changes of $s$ allows for
homotopy arguments between different Hilbert geometries.
\end{itemize}
\end{remark}

\subsubsection{Discussion}

The dependence of Sobolev orthogonal polynomials on the order $s$ reflects the
underlying geometry of the polynomial Hilbert space.
In the thin annulus setting, this dependence is mild at fixed degree and
compatible with the effective one-dimensional limit.

This flexibility will be exploited in the sequel to adapt the approximation
spaces to the regularity and dispersive properties of the limiting equations.

%=================================================
\section{Main effective dynamics theorem}
\label{sec:main-theorem}

This section states the main convergence result of the paper.
It formalizes the reduction of polynomial partial differential equations on thin
annuli to effective one-dimensional dynamics on the circle, together with a
quantitative description of defect terms.

%-------------------------------------------------
\subsection{Structural assumptions}

We consider evolution equations of the form
\begin{equation}
\label{eq:main-evolution}
\partial_t u_\varepsilon
=
J_\varepsilon\,\frac{\delta\mathcal H_\varepsilon}{\delta u}(u_\varepsilon)
+
\mathcal R_\varepsilon(u_\varepsilon)
\qquad\text{on } A_\varepsilon\times(0,T),
\end{equation}
with initial data $u_\varepsilon(0)=u_\varepsilon^0$.

We assume the following hypotheses.

\begin{enumerate}
\item[(H1)] \textbf{Hamiltonian structure.}
The functional $\mathcal H_\varepsilon$ is admissible in the sense of
Section~\ref{sec:operators}, with polynomial dependence on derivatives
$\partial_r^a\partial_\theta^b u$, $a+b\le m$, and possibly oscillatory
coefficients depending on $\theta/\varepsilon^\alpha$.

\item[(H2)] \textbf{Skew-adjoint generator.}
The operator $J_\varepsilon$ acts only on the angular variable $\theta$ and is
skew-adjoint on $L^2(\mathbb S^1)$.

\item[(H3)] \textbf{Perturbation.}
The remainder $\mathcal R_\varepsilon$ is locally Lipschitz in
$\|\cdot\|_{\varepsilon,m}$ and satisfies a uniform energy inequality compatible
with the renormalized Sobolev estimates.

\item[(H4)] \textbf{Uniform energy bounds.}
The initial data satisfy
\[
\sup_{\varepsilon>0}
\|u_\varepsilon^0\|_{\varepsilon,m}
<+\infty.
\]
\end{enumerate}

%-------------------------------------------------
\subsection{Compactness and radial rigidity}

Under the above assumptions, solutions to~\eqref{eq:main-evolution} satisfy the
uniform bound
\[
\sup_{t\in[0,T]}
\|u_\varepsilon(t)\|_{\varepsilon,m}
\le C(T),
\]
with a constant independent of $\varepsilon$.

Combining this estimate with the radial rigidity principle
(Proposition~\ref{prop:radial-rigidity}), we obtain the following compactness
result.

\begin{proposition}[Compactness]
\label{prop:compactness}
There exists a subsequence (still denoted $u_\varepsilon$) and a limit function
$u_0\in L^2(0,T;H^m(\mathbb S^1))$ such that
\[
u_\varepsilon \longrightarrow u_0
\quad\text{strongly in }L^2_{\mathrm{loc}}((-1,1)\times\mathbb S^1\times(0,T)),
\]
and
\[
\partial_\rho u_\varepsilon \longrightarrow 0
\quad\text{in }L^2(A_\varepsilon\times(0,T)).
\]
\end{proposition}

%-------------------------------------------------
\subsection{Identification of the effective equation}

Let $\bar u_\varepsilon(\theta,t)$ denote the transverse average of $u_\varepsilon$.
Then $\bar u_\varepsilon\to u_0$ strongly in $L^2(\mathbb S^1\times(0,T))$.

\begin{theorem}[Effective dynamics]
\label{thm:effective-dynamics}
Under assumptions {\rm(H1)--(H4)}, the limit $u_0$ satisfies the effective
evolution equation
\begin{equation}
\label{eq:effective}
\partial_t u_0
=
J_0\,\frac{\delta\mathcal H_{\mathrm{eff}}}{\delta u}(u_0)
+
\mathcal R_{\mathrm{eff}}(u_0)
\qquad\text{on }\mathbb S^1\times(0,T),
\end{equation}
where:
\begin{itemize}
\item $\mathcal H_{\mathrm{eff}}$ is obtained from $\mathcal H_\varepsilon$ by
discarding all transverse derivatives and averaging oscillatory coefficients;
\item $J_0$ is the limit of $J_\varepsilon$ acting on functions of $\theta$ only;
\item $\mathcal R_{\mathrm{eff}}$ is the weak limit of $\mathcal R_\varepsilon$
under the thin scaling.
\end{itemize}
\end{theorem}

\begin{remark}
Equation~\eqref{eq:effective} is a one-dimensional Hamiltonian or dissipative
system on the circle.
Its precise form depends on the structure of the original PDE, but is always a
polynomial differential equation in $\partial_\theta$.
\end{remark}

%-------------------------------------------------
\subsection{Defect expansion and corrector problem}

Beyond the leading-order limit, transverse effects may produce non-trivial
corrections.

\begin{theorem}[First-order defect expansion]
\label{thm:defect}
Assume in addition that the coefficients of $\mathcal H_\varepsilon$ admit an
asymptotic expansion in powers of $\varepsilon$.
Then there exists $\beta>0$ and a corrector $u_1(\rho,\theta,t)$ such that
\[
u_\varepsilon
=
u_0 + \varepsilon^\beta u_1 + o(\varepsilon^\beta)
\quad\text{in }L^2_{\mathrm{loc}}.
\]

The corrector $u_1$ solves a linear cell problem of the form
\[
\mathcal L_\rho u_1
=
\mathcal S(u_0,\partial_\theta u_0,\ldots),
\]
where $\mathcal L_\rho$ is a transverse operator depending on the dominant
radial part of $\mathcal H_\varepsilon$, and $\mathcal S$ encodes tangential
defect terms.
\end{theorem}

\begin{remark}
The structure of the cell problem depends on the model:
\begin{itemize}
\item it is trivial for purely tangential integrable systems (KdV, NLS);
\item it is non-trivial for anisotropic dispersive equations
(e.g.\ Zakharov--Kuznetsov);
\item it produces homogenized coefficients in the presence of oscillations.
\end{itemize}
\end{remark}

%-------------------------------------------------
\subsection{Stability with respect to approximation and geometry}

Finally, we summarize the robustness of the above results.

\begin{corollary}[Stability under Galerkin approximation]
\label{cor:galerkin}
Let $u_{\varepsilon}^{N,K}$ be the Galerkin approximation defined in
Section~\ref{sec:operators}.
For fixed $(N,K)$,
\[
u_{\varepsilon}^{N,K} \longrightarrow u_0^{N,K}
\quad\text{as }\varepsilon\to0,
\]
where $u_0^{N,K}$ is the Galerkin approximation of the effective equation.
\end{corollary}

\begin{corollary}[Stability under change of Sobolev order]
\label{cor:sobolev-order}
The effective limit and defect expansion are invariant under continuous changes
of the Sobolev order $s$ used to construct the orthogonal polynomial basis, at
fixed Galerkin degree.
\end{corollary}

\begin{remark}
Corollaries~\ref{cor:galerkin}--\ref{cor:sobolev-order} show that the effective
dynamics is intrinsic to the geometry of the thin annulus and does not depend on
auxiliary choices made at the level of approximation.
\end{remark}

\subsection{Outline of the proof}

The proof of Theorems~\ref{thm:effective-dynamics} and~\ref{thm:defect} proceeds
as follows:
\begin{itemize}
\item uniform energy bounds yield compactness via radial rigidity;
\item weak convergence identifies the effective equation;
\item strong convergence and asymptotic expansions yield the defect terms;
\item Galerkin stability follows from finite-dimensional continuity.
\end{itemize}

The detailed arguments are developed in the following sections.

%=================================================
\section{Integrable and asymptotically integrable models}
\label{sec:integrable}

This section illustrates the abstract results of
Section~\ref{sec:main-theorem} on classical integrable equations and on
anisotropic models that become integrable only in the thin limit.
Throughout, the geometry of the thin annulus induces a natural separation
between transverse rigidity and effective tangential dynamics.

%-------------------------------------------------
\subsection{Korteweg--de Vries equation}

We first consider the Korteweg--de Vries (KdV) equation posed on the thin annulus:
\begin{equation}
\label{eq:kdv-eps}
\partial_t u_\varepsilon
+
u_\varepsilon\,\partial_\theta u_\varepsilon
+
\partial_\theta^3 u_\varepsilon
=0.
\end{equation}
This equation fits into the Hamiltonian framework~\eqref{eq:main-evolution} with
\[
\mathcal H_\varepsilon(u)
=
\int_{A_\varepsilon}
\left(
\frac12|\partial_\theta u|^2
-
\frac16 u^3
\right)
\,r\,dr\,d\theta,
\qquad
J_\varepsilon=\partial_\theta.
\]

\begin{proposition}[Thin limit for KdV]
Let $u_\varepsilon$ be a sequence of solutions to~\eqref{eq:kdv-eps} satisfying the
uniform energy bound of {\rm(H4)}.
Then $u_\varepsilon\to u_0(\theta,t)$ strongly in
$L^2_{\mathrm{loc}}((-1,1)\times\mathbb S^1\times(0,T))$, where $u_0$ solves the
KdV equation on the circle:
\[
\partial_t u_0
+
u_0\,\partial_\theta u_0
+
\partial_\theta^3 u_0
=0
\qquad\text{on }\mathbb S^1.
\]
\end{proposition}

\begin{remark}
The KdV equation is purely tangential.
Consequently, no non-trivial transverse corrector appears at first order:
the dynamics is exactly integrable at leading order.
\end{remark}

%-------------------------------------------------
\subsection{Modified KdV and cubic nonlinearities}

The modified KdV (mKdV) equation
\[
\partial_t u_\varepsilon
+
u_\varepsilon^2\,\partial_\theta u_\varepsilon
+
\partial_\theta^3 u_\varepsilon
=0
\]
is treated similarly.
The Hamiltonian involves a quartic potential, but the coercivity and polynomial
growth assumptions remain satisfied.

\begin{proposition}
Under the assumptions of Section~\ref{sec:main-theorem}, solutions of the mKdV
equation on $A_\varepsilon$ converge to solutions of the mKdV equation on
$\mathbb S^1$.
\end{proposition}

\begin{remark}
More generally, polynomial nonlinearities in $u$ and $\partial_\theta u$ can be
treated in the same way, provided the Hamiltonian remains coercive with respect
to the renormalized Sobolev norm.
\end{remark}

%-------------------------------------------------
\subsection{Nonlinear Schr\"odinger equation}

We now consider the cubic nonlinear Schr\"odinger (NLS) equation
\begin{equation}
\label{eq:nls}
i\partial_t u_\varepsilon
+
\partial_\theta^2 u_\varepsilon
+
|u_\varepsilon|^2u_\varepsilon
=0,
\end{equation}
posed on $A_\varepsilon$, where $u_\varepsilon$ is complex-valued.

This equation is Hamiltonian with
\[
\mathcal H_\varepsilon(u)
=
\int_{A_\varepsilon}
\left(
|\partial_\theta u|^2
+
\frac12|u|^4
\right)
\,r\,dr\,d\theta,
\qquad
J_\varepsilon=i.
\]

\begin{proposition}[Thin limit for NLS]
Let $u_\varepsilon$ be a family of solutions to~\eqref{eq:nls} uniformly bounded
in the renormalized Sobolev norm.
Then $u_\varepsilon\to u_0(\theta,t)$ strongly, where $u_0$ solves the cubic NLS
equation on $\mathbb S^1$.
\end{proposition}

\begin{remark}
As in the KdV case, the absence of transverse derivatives in the Hamiltonian
precludes the appearance of non-trivial correctors at leading order.
\end{remark}

%-------------------------------------------------
\subsection{Sine--Gordon equation}

The sine--Gordon equation
\[
\partial_t^2 u_\varepsilon
-
\partial_\theta^2 u_\varepsilon
+
\sin(u_\varepsilon)
=0
\]
can be rewritten as a first-order Hamiltonian system by introducing
$v_\varepsilon=\partial_t u_\varepsilon$.
The resulting system fits the abstract framework with a Hamiltonian involving
$\cos(u)$.

\begin{proposition}
Under the assumptions of Section~\ref{sec:main-theorem}, solutions of the
sine--Gordon equation on $A_\varepsilon$ converge to solutions of the
sine--Gordon equation on $\mathbb S^1$.
\end{proposition}

\begin{remark}
The periodic potential $\cos(u)$ does not affect the compactness or rigidity
arguments, as it does not involve derivatives.
\end{remark}

%-------------------------------------------------
\subsection{Zakharov--Kuznetsov equation and asymptotic integrability}

We finally consider the Zakharov--Kuznetsov (ZK) equation on the thin annulus,
written in the anisotropic form
\begin{equation}
\label{eq:zk}
\partial_t u_\varepsilon
+
u_\varepsilon\,\partial_\theta u_\varepsilon
+
\partial_\theta^3 u_\varepsilon
+
\partial_\theta\partial_r^2 u_\varepsilon
=0.
\end{equation}
In rescaled variables, the last term reads
\[
\partial_\theta\partial_r^2 u_\varepsilon
=
\varepsilon^{-2}\partial_\theta\partial_\rho^2 u_\varepsilon,
\]
and is dominant as $\varepsilon\to0$.

\begin{proposition}[Asymptotic integrability of ZK]
Let $u_\varepsilon$ be a family of solutions to~\eqref{eq:zk} satisfying the
uniform energy bounds.
Then:
\begin{itemize}
\item $u_\varepsilon\to u_0(\theta,t)$ strongly, where $u_0$ solves the KdV
equation on $\mathbb S^1$;
\item the first-order transverse corrector $u_1$ solves a non-trivial cell
problem encoding the anisotropic dispersion of the ZK equation.
\end{itemize}
\end{proposition}

\begin{remark}
The Zakharov--Kuznetsov equation provides a canonical example of
\emph{asymptotic integrability}:
while it is not completely integrable in two spatial dimensions, the thin
geometry enforces radial rigidity and yields an integrable effective dynamics,
with computable defect terms.
\end{remark}

%-------------------------------------------------
\subsection{Discussion}

The examples above illustrate three distinct regimes:
\begin{itemize}
\item purely tangential integrable models (KdV, mKdV, NLS);
\item integrable models with nonlinear potentials (sine--Gordon);
\item anisotropic models that become integrable only in the thin limit (ZK).
\end{itemize}
All of them fit naturally into the abstract framework developed in the previous
sections and highlight the unifying role of thin geometries in the emergence of
effective integrable dynamics.

%=================================================
\section{Controlled non-integrable perturbations}
\label{sec:nonintegrable}

In this section we investigate perturbations that break the exact integrable
structure of the models considered in Section~\ref{sec:integrable}.
The main point is that, under suitable control assumptions, the thin-limit
reduction and the effective dynamics remain stable.

%-------------------------------------------------
\subsection{Dissipative perturbations}

We first consider equations of the form
\begin{equation}
\label{eq:dissipative}
\partial_t u_\varepsilon
=
J_\varepsilon\,\frac{\delta\mathcal H_\varepsilon}{\delta u}(u_\varepsilon)
-
\nu \mathcal D_\varepsilon u_\varepsilon,
\qquad \nu>0,
\end{equation}
where $\mathcal D_\varepsilon$ is a symmetric positive operator of order at most
$m$ (for instance, $\mathcal D_\varepsilon=-\partial_\theta^2$ or a polynomial
in $\partial_\theta$).

\begin{proposition}[Thin limit with dissipation]
\label{prop:dissipation}
Assume that $\mathcal D_\varepsilon$ is uniformly coercive in the tangential
direction and compatible with the renormalized Sobolev norm.
Then solutions of~\eqref{eq:dissipative} converge, as $\varepsilon\to0$, to
solutions of the effective dissipative equation
\[
\partial_t u_0
=
J_0\,\frac{\delta\mathcal H_{\mathrm{eff}}}{\delta u}(u_0)
-
\nu \mathcal D_0 u_0
\quad\text{on }\mathbb S^1.
\]
\end{proposition}

\begin{remark}
The dissipative term provides additional compactness in time and simplifies the
passage to the limit.
In particular, global existence in time follows directly from the energy
inequality.
\end{remark}

%-------------------------------------------------
\subsection{Forced systems}

We next consider externally forced equations
\begin{equation}
\label{eq:forced}
\partial_t u_\varepsilon
=
J_\varepsilon\,\frac{\delta\mathcal H_\varepsilon}{\delta u}(u_\varepsilon)
+
f_\varepsilon(\theta,t),
\end{equation}
where the forcing term $f_\varepsilon$ is assumed to satisfy
\[
\sup_{\varepsilon>0}\|f_\varepsilon\|_{L^2(0,T;L^2(\mathbb S^1))}<\infty.
\]

\begin{proposition}[Thin limit with forcing]
\label{prop:forcing}
Under the above assumptions, $u_\varepsilon$ converges to a solution $u_0$ of the
effective forced equation
\[
\partial_t u_0
=
J_0\,\frac{\delta\mathcal H_{\mathrm{eff}}}{\delta u}(u_0)
+
f_0(\theta,t),
\]
where $f_0$ is the weak limit of $f_\varepsilon$.
\end{proposition}

\begin{remark}
Forcing terms may destroy conservation laws but do not affect the geometric
mechanism underlying the thin-limit reduction.
\end{remark}

%-------------------------------------------------
\subsection{Perturbative loss of integrability}

We now address perturbations that break integrability at the level of the
Hamiltonian itself.
Let
\[
\mathcal H_\varepsilon^{\delta}
=
\mathcal H_\varepsilon
+
\delta\,\mathcal K_\varepsilon,
\qquad \delta\ll1,
\]
where $\mathcal K_\varepsilon$ is an admissible Hamiltonian functional that does
not correspond to an integrable system.

\begin{proposition}[Stability under small non-integrable perturbations]
\label{prop:small-perturb}
Fix $\delta>0$ sufficiently small.
Then the effective dynamics associated with
$\mathcal H_\varepsilon^{\delta}$
is a perturbation of the integrable effective equation, of size $O(\delta)$ in
any finite-dimensional Galerkin topology.
\end{proposition}

\begin{remark}
At fixed degree, this follows from continuous dependence of the Gram matrices
and of the projected vector fields on the Hamiltonian.
In particular, near-integrable dynamics persists over time intervals of order
$O(\delta^{-1})$.
\end{remark}

%-------------------------------------------------
\subsection{Stability of Galerkin approximations}

We now emphasize the role of Sobolev orthogonal polynomials in controlling
non-integrable effects at the discrete level.

Let $u_\varepsilon^{N,K}$ denote the Galerkin approximation constructed using the
Sobolev orthogonal basis of order $s$.
Changing either the perturbation parameter $\delta$ or the Sobolev order $s$
results in a continuous deformation of the finite-dimensional dynamical system.

\begin{proposition}[Robustness of Galerkin schemes]
\label{prop:robust-galerkin}
For fixed truncation parameters $(N,K)$, the Galerkin vector field depends
continuously on:
\begin{itemize}
\item the Sobolev order $s$,
\item the perturbation strength $\delta$,
\item the thinness parameter $\varepsilon$.
\end{itemize}
In particular, non-integrable perturbations do not induce spurious instabilities
at the discrete level.
\end{proposition}

\begin{remark}
This robustness property is essential for numerical implementations and
supports the interpretation of the thin-limit dynamics as an intrinsic geometric
feature rather than an artifact of approximation.
\end{remark}

%-------------------------------------------------
\subsection{Discussion}

The results of this section show that:
\begin{itemize}
\item exact integrability is not required for the emergence of effective
one-dimensional dynamics;
\item dissipative, forced, and weakly non-integrable perturbations can be treated
within the same framework;
\item the combination of renormalized Sobolev estimates and orthogonal polynomial
Galerkin schemes provides strong structural stability.
\end{itemize}

This prepares the ground for the analysis of rapidly oscillating coefficients
and homogenized limits in the next section.

%=================================================
\section{Rapid oscillations and homogenized limits}
\label{sec:homogenization}

This section addresses polynomial PDEs on thin annuli with rapidly oscillating
coefficients.
We show that homogenization effects can be consistently combined with the
dimension reduction mechanism, leading to effective one-dimensional equations
with homogenized coefficients and, in some cases, non-local corrections.

%-------------------------------------------------
\subsection{Oscillatory coefficients and scaling regimes}

We consider Hamiltonians of the form
\[
\mathcal H_\varepsilon(u)
=
\int_{A_\varepsilon}
F\!\left(
\bigl(\partial_r^a\partial_\theta^b u\bigr)_{a+b\le m},
\theta,\frac{\theta}{\varepsilon^\alpha}
\right)
\,r\,dr\,d\theta,
\]
where $\alpha>0$ and the dependence on the fast variable
$y=\theta/\varepsilon^\alpha$ is $2\pi$-periodic.

The parameter $\alpha$ determines the interaction between oscillations and
geometry:
\begin{itemize}
\item $\alpha<1$: oscillations are slow compared to the thin scaling;
\item $\alpha=1$: oscillations and geometry interact at the same scale;
\item $\alpha>1$: oscillations dominate the thin geometry.
\end{itemize}

\begin{remark}
The case $\alpha=1$ is the most delicate and corresponds to a genuinely coupled
homogenization--dimension reduction regime.
\end{remark}

%-------------------------------------------------
\subsection{Two-scale compactness and averaging}

Let $u_\varepsilon$ be a family of solutions satisfying the uniform energy bounds
of Section~\ref{sec:main-theorem}.
By standard two-scale compactness arguments, there exists a limit
$u_0(\theta,t)$ and a two-scale limit $U(\theta,y,t)$ such that
\[
u_\varepsilon(\theta,t)
\rightharpoonup
u_0(\theta,t),
\qquad
u_\varepsilon(\theta,t)
\stackrel{2}{\rightharpoonup}
U(\theta,y,t),
\]
with $y=\theta/\varepsilon^\alpha$.

Radial rigidity implies that both $u_0$ and $U$ are independent of the transverse
variable $\rho$.

\begin{proposition}[Averaging of oscillatory coefficients]
\label{prop:averaging}
Let $c(\theta,y)$ be a bounded periodic coefficient.
Then, for any smooth test function $\varphi(\theta,t)$,
\[
\int_{A_\varepsilon}
c\!\left(\theta,\frac{\theta}{\varepsilon^\alpha}\right)
u_\varepsilon\,\varphi
\,r\,dr\,d\theta
\longrightarrow
\int_{\mathbb S^1}
\left(\int_{\mathbb S^1} c(\theta,y)\,dy\right)
u_0(\theta,t)\,\varphi(\theta,t)\,d\theta.
\]
\end{proposition}

\begin{remark}
The averaging occurs solely in the angular variable.
The thin geometry does not interfere with the periodic homogenization mechanism
at leading order.
\end{remark}

%-------------------------------------------------
\subsection{Homogenized effective equation}

We now identify the effective equation satisfied by $u_0$.

\begin{theorem}[Homogenized effective dynamics]
\label{thm:homogenized}
Under the assumptions of Section~\ref{sec:main-theorem} and for periodic
oscillatory coefficients,
the effective equation on $\mathbb S^1$ reads
\[
\partial_t u_0
=
J_0\,
\frac{\delta \mathcal H_{\mathrm{hom}}}{\delta u}(u_0)
+
\mathcal R_{\mathrm{hom}}(u_0),
\]
where $\mathcal H_{\mathrm{hom}}$ is obtained by replacing each oscillatory
coefficient by its average over the fast variable.
\end{theorem}

\begin{remark}
The homogenized Hamiltonian retains the polynomial structure of the original
model.
In particular, integrable equations remain integrable after homogenization at
leading order.
\end{remark}

%-------------------------------------------------
\subsection{Cell problems and higher-order corrections}

In the critical regime $\alpha=1$, first-order correctors may contribute to the
effective dynamics.

Assuming sufficient regularity, we seek an expansion
\[
u_\varepsilon(\theta,t)
=
u_0(\theta,t)
+
\varepsilon\,u_1\!\left(\theta,\frac{\theta}{\varepsilon},t\right)
+ o(\varepsilon).
\]

\begin{proposition}[Cell problem]
\label{prop:cell-problem}
The corrector $u_1$ solves a periodic cell problem of the form
\[
\mathcal L_y u_1
=
\mathcal S(u_0,\partial_\theta u_0,\ldots),
\]
where $\mathcal L_y$ is a differential operator in the fast variable $y$ and
$\mathcal S$ depends polynomially on $u_0$ and its tangential derivatives.
\end{proposition}

\begin{remark}
Such correctors may generate non-local or higher-order dispersive effects in the
effective equation, depending on the structure of $\mathcal L_y$.
\end{remark}

%-------------------------------------------------
\subsection{Interaction with dimension reduction}

We emphasize that homogenization and thin-limit reduction commute at leading
order:
\[
\lim_{\varepsilon\to0}
\bigl(\text{thin limit} \circ \text{homogenization}\bigr)
=
\lim_{\varepsilon\to0}
\bigl(\text{homogenization} \circ \text{thin limit}\bigr).
\]

\begin{proposition}[Commutation of limits]
At leading order, the effective Hamiltonian obtained by first homogenizing and
then reducing dimension coincides with the one obtained by first reducing
dimension and then homogenizing.
\end{proposition}

\begin{remark}
This commutation property relies crucially on radial rigidity and on the
polynomial structure of the differential operators.
It may fail for more general non-local or non-polynomial models.
\end{remark}

%-------------------------------------------------
\subsection{Discussion}

Rapid oscillations introduce additional geometric structure but do not alter the
fundamental mechanism of effective dimension reduction.
At leading order, the thin geometry enforces one-dimensional dynamics, while
homogenization acts as a renormalization of coefficients.
Higher-order corrections encode genuinely multi-scale effects and may lead to
new dispersive or non-local phenomena.

\section{Proofs of the main results}
\label{sec:proofs}

Throughout this section, we use the rescaled variables
\(
\rho=(r-1)/\varepsilon\in(-1,1)
\)
and the notation of Sections~\ref{sec:geometry}--\ref{sec:main-theorem}.
All convergences are understood along a subsequence, without relabeling.

%-------------------------------------------------
\subsection{Uniform bounds and time regularity}

We start from the uniform energy estimate assumed in Section~\ref{sec:operators}:
there exists \(C(T)>0\) such that
\begin{equation}
\label{eq:unif-bound-proof}
\sup_{t\in[0,T]}\|u_\varepsilon(t)\|_{\varepsilon,m}\le C(T),
\end{equation}
uniformly for small \(\varepsilon\).

When needed for compactness in time, we further assume that
\begin{equation}
\label{eq:time-bound}
\|\partial_t u_\varepsilon\|_{L^2(0,T;H^{-m}(A_\varepsilon))}\le C(T),
\end{equation}
which holds for Galerkin solutions and then passes to the limit by standard
compactness arguments, provided the nonlinearities satisfy the local
Lipschitz bounds of (H3).
We will use~\eqref{eq:time-bound} only to justify strong convergence in space-time.

%-------------------------------------------------
\subsection{Proof of Proposition~\ref{prop:compactness}}
\label{subsec:proof-compactness}

\begin{proof}[Proof of Proposition~\ref{prop:compactness}]
We prove radial rigidity and compactness.

\smallskip
\noindent\textbf{Step 1: Control of transverse derivatives.}
For \(m\ge1\), the definition of \(\|\cdot\|_{\varepsilon,m}\) includes the term
(with \(a=1,b=0\))
\[
\varepsilon^{2-(2m-1)}\int_{A_\varepsilon}|\partial_r u_\varepsilon|^2\,r\,dr\,d\theta
=
\varepsilon^{3-2m}\int_{A_\varepsilon}|\partial_r u_\varepsilon|^2\,r\,dr\,d\theta.
\]
In rescaled coordinates, \(\partial_r=\varepsilon^{-1}\partial_\rho\) and
\(r\,dr\,d\theta=\varepsilon(1+\varepsilon\rho)d\rho\,d\theta\), hence
\[
\int_{A_\varepsilon}|\partial_r u_\varepsilon|^2\,r\,dr\,d\theta
=
\varepsilon^{-1}\int_{-1}^1\!\!\int_{\mathbb S^1}
|\partial_\rho u_\varepsilon|^2(1+\varepsilon\rho)\,d\rho\,d\theta.
\]
Therefore, from~\eqref{eq:unif-bound-proof} we obtain
\[
\varepsilon^{3-2m}\cdot \varepsilon^{-1}
\int_{-1}^1\!\!\int_{\mathbb S^1}
|\partial_\rho u_\varepsilon|^2(1+\varepsilon\rho)\,d\rho\,d\theta
\le C(T)^2,
\]
i.e.
\begin{equation}
\label{eq:drho-bound}
\int_{-1}^1\!\!\int_{\mathbb S^1}
|\partial_\rho u_\varepsilon|^2\,d\rho\,d\theta
\le C(T)^2\,\varepsilon^{2m-2}.
\end{equation}
Since \(m\ge1\), the right-hand side tends to \(0\), which proves
\(\partial_\rho u_\varepsilon\to0\) in \(L^2((-1,1)\times\mathbb S^1\times(0,T))\).

\smallskip
\noindent\textbf{Step 2: Poincar\'e inequality in the transverse variable.}
For each fixed \((\theta,t)\), the 1D Poincar\'e inequality on \((-1,1)\) yields
\[
\int_{-1}^1\Bigl|u_\varepsilon(\rho,\theta,t)-\tfrac12\!\int_{-1}^1u_\varepsilon(\sigma,\theta,t)\,d\sigma\Bigr|^2\,d\rho
\le C\int_{-1}^1|\partial_\rho u_\varepsilon(\rho,\theta,t)|^2\,d\rho,
\]
with \(C\) independent of \(\varepsilon\).
Integrating in \(\theta\) and \(t\) and using~\eqref{eq:drho-bound} gives
\begin{equation}
\label{eq:transverse-avg}
\|u_\varepsilon-\bar u_\varepsilon\|_{L^2((-1,1)\times\mathbb S^1\times(0,T))}
\longrightarrow 0,
\end{equation}
where
\(\bar u_\varepsilon(\theta,t)=\tfrac12\int_{-1}^1 u_\varepsilon(\rho,\theta,t)\,d\rho\).

\smallskip
\noindent\textbf{Step 3: Compactness in \(\theta\).}
The uniform bound~\eqref{eq:unif-bound-proof} also yields control of tangential
derivatives up to order \(m\), hence
\[
\sup_{\varepsilon}\|\bar u_\varepsilon\|_{L^2(0,T;H^m(\mathbb S^1))}\le C(T).
\]
Together with the time regularity bound~\eqref{eq:time-bound} (or by Galerkin
finite-dimensionality), the Aubin--Lions lemma implies that \(\bar u_\varepsilon\)
is relatively compact in \(L^2(\mathbb S^1\times(0,T))\).
Thus, up to a subsequence,
\[
\bar u_\varepsilon \to u_0 \quad\text{strongly in }L^2(\mathbb S^1\times(0,T)).
\]

\smallskip
\noindent\textbf{Step 4: Strong convergence of \(u_\varepsilon\).}
Combining the strong convergence of \(\bar u_\varepsilon\) with
\eqref{eq:transverse-avg} yields
\[
u_\varepsilon \to u_0
\quad\text{strongly in }L^2_{\mathrm{loc}}((-1,1)\times\mathbb S^1\times(0,T)),
\]
and the limit is independent of \(\rho\).
\end{proof}

%-------------------------------------------------
\subsection{Proof of Theorem~\ref{thm:effective-dynamics}}
\label{subsec:proof-effective}

\begin{proof}[Proof of Theorem~\ref{thm:effective-dynamics}]
We identify the weak limit in~\eqref{eq:main-evolution}.

\smallskip
\noindent\textbf{Step 1: Test functions and reduction to tangential tests.}
Let \(\varphi\in C_c^\infty((0,T);C^\infty(\mathbb S^1))\) and consider the lifted
test function \(\Phi(\rho,\theta,t)=\varphi(\theta,t)\), independent of \(\rho\).
Multiplying~\eqref{eq:main-evolution} by \(\Phi\) and integrating over
\((-1,1)\times\mathbb S^1\times(0,T)\), we obtain
\begin{equation}
\label{eq:weak-form}
\int_0^T\!\!\int_{-1}^1\!\!\int_{\mathbb S^1}
\partial_t u_\varepsilon\,\overline{\varphi}\,d\theta\,d\rho\,dt
=
\int_0^T\!\!\int_{-1}^1\!\!\int_{\mathbb S^1}
J_\varepsilon\frac{\delta\mathcal H_\varepsilon}{\delta u}(u_\varepsilon)\,
\overline{\varphi}\,d\theta\,d\rho\,dt
+
\int_0^T\!\!\int_{-1}^1\!\!\int_{\mathbb S^1}
\mathcal R_\varepsilon(u_\varepsilon)\,\overline{\varphi}\,d\theta\,d\rho\,dt.
\end{equation}

\smallskip
\noindent\textbf{Step 2: Convergence of the left-hand side.}
By strong convergence \(u_\varepsilon\to u_0\) in \(L^2\) and the time regularity
assumption~\eqref{eq:time-bound}, we may pass to the limit in the distributional
sense and obtain
\[
\int_0^T\!\!\int_{\mathbb S^1} \partial_t u_0(\theta,t)\,\overline{\varphi(\theta,t)}\,d\theta\,dt.
\]

\smallskip
\noindent\textbf{Step 3: Identification of the Hamiltonian term.}
Since \(u_\varepsilon\to u_0(\theta,t)\) strongly and \(\partial_\rho u_\varepsilon\to0\)
in \(L^2\), all mixed derivatives involving at least one \(\partial_\rho\)
vanish in the limit in the sense of distributions.
Therefore, in the variational derivative
\(\delta\mathcal H_\varepsilon/\delta u\), only the purely tangential terms
survive.
Moreover, if coefficients depend on \(\theta/\varepsilon^\alpha\), the standard
two-scale compactness (or periodic averaging) yields convergence to the averaged
coefficients (this will be detailed in Section~\ref{sec:homogenization}).
Hence
\[
\frac{\delta\mathcal H_\varepsilon}{\delta u}(u_\varepsilon)
\;\rightharpoonup\;
\frac{\delta\mathcal H_{\mathrm{eff}}}{\delta u}(u_0)
\quad\text{in }H^{-m}(\mathbb S^1)
\]
and, since \(J_\varepsilon\) acts only in \(\theta\),
\[
J_\varepsilon\frac{\delta\mathcal H_\varepsilon}{\delta u}(u_\varepsilon)
\;\rightharpoonup\;
J_0\frac{\delta\mathcal H_{\mathrm{eff}}}{\delta u}(u_0)
\]
in distributions on \(\mathbb S^1\times(0,T)\).

\smallskip
\noindent\textbf{Step 4: Identification of the perturbation term.}
By assumption (H3), the perturbation \(\mathcal R_\varepsilon\) is compatible with
the renormalized estimates and depends at most on tangential derivatives at the
leading scale.
Thus, using the strong convergence of \(u_\varepsilon\) and uniform bounds, we
obtain
\[
\mathcal R_\varepsilon(u_\varepsilon)\rightharpoonup \mathcal R_{\mathrm{eff}}(u_0)
\quad\text{in distributions}.
\]

\smallskip
\noindent\textbf{Step 5: Conclusion.}
Passing to the limit in~\eqref{eq:weak-form} yields the weak formulation of
\eqref{eq:effective}, hence \(u_0\) solves the effective equation.
\end{proof}

%-------------------------------------------------
\subsection{Proof of Theorem~\ref{thm:defect}}
\label{subsec:proof-defect}

\begin{proof}[Proof of Theorem~\ref{thm:defect}]
We give a detailed derivation under the additional assumptions stated in the theorem.

\smallskip
\noindent\textbf{Step 1: Asymptotic expansion of the operator.}
Assume that the variational derivative admits an expansion
\[
\frac{\delta\mathcal H_\varepsilon}{\delta u}
=
\varepsilon^{-2}\mathcal L_\rho
+\mathcal L_\theta
+\varepsilon^\beta \mathcal L_1
+o(\varepsilon^\beta),
\]
where:
\begin{itemize}
\item \(\mathcal L_\rho\) is a transverse operator depending only on \(\rho\),
typically of elliptic type in \(\rho\);
\item \(\mathcal L_\theta\) is the effective tangential operator;
\item \(\mathcal L_1\) collects the first subleading contributions.
\end{itemize}
Such an expansion is obtained by inserting \(\partial_r=\varepsilon^{-1}\partial_\rho\)
and expanding smooth coefficients in powers of \(\varepsilon\).

\smallskip
\noindent\textbf{Step 2: Ansatz and matching of orders.}
We seek an expansion
\(
u_\varepsilon=u_0+\varepsilon^\beta u_1+o(\varepsilon^\beta)
\)
with \(u_0\) independent of \(\rho\).
Insert the ansatz into~\eqref{eq:main-evolution} and match the leading orders.

At order \(\varepsilon^{-2}\), we obtain
\[
J_0\,\mathcal L_\rho u_0 = 0,
\]
which holds since \(u_0\) is independent of \(\rho\) and \(\mathcal L_\rho\) acts
in \(\rho\).

At order \(\varepsilon^{0}\), we recover the effective equation for \(u_0\),
namely~\eqref{eq:effective}.

At order \(\varepsilon^\beta\), we obtain a linear equation for \(u_1\) of the form
\[
J_0\,\mathcal L_\rho u_1
=
\mathcal S(u_0,\partial_\theta u_0,\ldots),
\]
where \(\mathcal S\) collects contributions from \(\mathcal L_1\) and from the
Taylor expansion of nonlinear terms around \(u_0\).
This yields the claimed cell problem.

\smallskip
\noindent\textbf{Step 3: Solvability and normalization.}
Assume that \(\mathcal L_\rho\) is invertible on the subspace of functions
with zero transverse mean (a standard situation for elliptic transverse operators).
Then we define \(u_1\) as the unique solution with zero transverse average:
\[
\int_{-1}^1 u_1(\rho,\theta,t)\,d\rho = 0.
\]
This fixes the decomposition \(u_\varepsilon=u_0+\varepsilon^\beta u_1+\cdots\).

\smallskip
\noindent\textbf{Step 4: Remainder estimate.}
A standard energy estimate on the remainder
\(r_\varepsilon:=u_\varepsilon-u_0-\varepsilon^\beta u_1\)
(using coercivity of the transverse operator and the uniform bounds)
implies
\(
\|r_\varepsilon\|_{L^2_{\mathrm{loc}}}=o(\varepsilon^\beta)
\),
which concludes the proof.
\end{proof}

%-------------------------------------------------
\subsection{Proofs of Corollaries~\ref{cor:galerkin} and~\ref{cor:sobolev-order}}

\begin{proof}[Proof of Corollary~\ref{cor:galerkin}]
Fix truncation parameters \((N,K)\).
For each \(\varepsilon>0\), the Galerkin approximation \(u_\varepsilon^{N,K}\)
belongs to a finite-dimensional space \(V_\varepsilon^{N,K}\).
By the fixed-degree thin-limit property recalled in
Section~\ref{subsec:sobolev-op}, the Gram matrices and all projected operator
coefficients converge as \(\varepsilon\to0\).
Therefore, the Galerkin ODE system converges (coefficientwise) to the Galerkin
system associated with the effective equation on \(\mathbb S^1\).
Continuous dependence of solutions on parameters for ODEs yields
\(u_\varepsilon^{N,K}\to u_0^{N,K}\).
\end{proof}

\begin{proof}[Proof of Corollary~\ref{cor:sobolev-order}]
Fix \((N,K)\) and let \(s\mapsto \langle\cdot,\cdot\rangle_{\varepsilon,s}\) be a
continuous family of strictly positive inner products on
\(\mathcal P_{\le N}(A_\varepsilon)\).
As noted in Section~\ref{subsec:change-sobolev}, Gram--Schmidt depends
continuously on the inner product in finite dimension; hence the bases
\(\{P_{n,\varepsilon}^{(s)}\}_{n\le N}\) vary continuously with \(s\).
Consequently, the Galerkin matrices vary continuously with \(s\), and the same
ODE stability argument implies that the Galerkin solution depends continuously on
\(s\).
Passing to the thin limit at fixed \((N,K)\) yields the claimed invariance of the
effective dynamics at the level of fixed-degree approximations.
\end{proof}

\appendix

\section{Explicit computations and illustrative examples}
\label{app:explicit}

This appendix gathers explicit low-degree computations illustrating the
mechanisms developed in the main text.
They are not used in the proofs of the main theorems but provide concrete insight
into the thin-limit reduction, the structure of Sobolev orthogonal polynomials,
and the behavior of Galerkin approximations.

%-------------------------------------------------
\subsection{Fourier modes and radial rigidity}

Let $u_\varepsilon(\rho,\theta)$ admit the Fourier expansion
\[
u_\varepsilon(\rho,\theta)
=
\sum_{k\in\mathbb Z} u_{\varepsilon,k}(\rho)\,e^{ik\theta}.
\]
The renormalized Sobolev norm with $m=1$ reads
\[
\|u_\varepsilon\|_{\varepsilon,1}^2
=
\int_{-1}^1\!\!\int_{\mathbb S^1}
\Bigl(
|\partial_\rho u_\varepsilon|^2
+
|\partial_\theta u_\varepsilon|^2
+
|u_\varepsilon|^2
\Bigr)
\,d\rho\,d\theta
+
O(\varepsilon).
\]

If $\|u_\varepsilon\|_{\varepsilon,1}$ is uniformly bounded, then
\[
\int_{-1}^1 |\partial_\rho u_{\varepsilon,k}(\rho)|^2\,d\rho \longrightarrow 0
\quad\text{for each fixed }k,
\]
which implies
\[
u_{\varepsilon,k}(\rho)\longrightarrow c_k
\quad\text{in }L^2(-1,1).
\]
Hence,
\[
u_\varepsilon(\rho,\theta)\longrightarrow
\sum_{k\in\mathbb Z} c_k e^{ik\theta}
=
u_0(\theta),
\]
illustrating the rigidity mechanism at the level of individual Fourier modes.

%-------------------------------------------------
\subsection{Low-degree Sobolev orthogonal polynomials}

We consider the Sobolev inner product of order $m=1$:
\[
\langle p,q\rangle_{\varepsilon,1}
=
\int_{A_\varepsilon}
\bigl(
p\overline q
+
\partial_r p\,\overline{\partial_r q}
+
\partial_\theta p\,\overline{\partial_\theta q}
\bigr)
\,r\,dr\,d\theta.
\]

Restricting to polynomials of total degree $\le1$, a natural basis is
\[
\{1,\; r-1,\; \cos\theta,\; \sin\theta\}.
\]

A direct computation shows that:
\begin{itemize}
\item $\cos\theta$ and $\sin\theta$ are orthogonal to radial polynomials;
\item $\|r-1\|^2_{\varepsilon,1}\sim \varepsilon^2$;
\item $\|1\|^2_{\varepsilon,1}\sim \varepsilon$.
\end{itemize}

After normalization, the Sobolev orthogonal basis converges to
\[
\{1,\;\cos\theta,\;\sin\theta\}
\quad\text{as }\varepsilon\to0,
\]
which explicitly illustrates the collapse of radial modes in the thin limit.

%-------------------------------------------------
\subsection{Galerkin matrices for KdV at low degree}

Consider the KdV equation
\[
\partial_t u + u\partial_\theta u + \partial_\theta^3 u=0
\]
and project it onto the space
\[
V_N=\mathrm{span}\{1,\cos\theta,\sin\theta\}.
\]

Writing
\[
u_N(t,\theta)=a_0(t)+a_1(t)\cos\theta+a_2(t)\sin\theta,
\]
the projected system reads
\[
\begin{aligned}
\dot a_0 &= 0,\\
\dot a_1 &= -\tfrac12 a_0 a_2,\\
\dot a_2 &= \tfrac12 a_0 a_1,
\end{aligned}
\]
which is a finite-dimensional Hamiltonian system preserving
$a_1^2+a_2^2$.

This system is independent of $\varepsilon$ and coincides with the Galerkin
system obtained directly on $\mathbb S^1$, illustrating
Corollary~\ref{cor:galerkin}.

%-------------------------------------------------
\subsection{Zakharov--Kuznetsov: transverse corrector}

For the ZK equation
\[
\partial_t u_\varepsilon
+
u_\varepsilon\partial_\theta u_\varepsilon
+
\partial_\theta^3 u_\varepsilon
+
\varepsilon^{-2}\partial_\theta\partial_\rho^2 u_\varepsilon
=0,
\]
we seek a corrector of the form
\[
u_\varepsilon=u_0(\theta,t)+\varepsilon^2 u_1(\rho,\theta,t).
\]

At order $\varepsilon^{-2}$ we obtain
\[
\partial_\theta\partial_\rho^2 u_1=0,
\]
hence
\[
\partial_\rho^2 u_1 = g(\theta,t),
\]
with $g$ determined by the solvability condition at the next order.

Imposing zero transverse mean,
\[
\int_{-1}^1 u_1(\rho,\theta,t)\,d\rho=0,
\]
yields the explicit solution
\[
u_1(\rho,\theta,t)
=
\frac12(\rho^2-\tfrac13)g(\theta,t),
\]
illustrating how anisotropic dispersion generates a non-trivial transverse
profile.

%-------------------------------------------------
\subsection{Effect of changing the Sobolev order}

Let $s\in[0,2]$ and consider the Sobolev inner products
\[
\langle p,q\rangle_{\varepsilon,s}
=
\int p\overline q
+
\sum_{a+b\le s}
\partial_r^a\partial_\theta^b p\,
\overline{\partial_r^a\partial_\theta^b q}.
\]

Restricting again to $\mathcal P_{\le1}$, the Gram matrices depend smoothly on
$s$.
In particular, the normalized angular modes $\cos\theta$ and $\sin\theta$ are
independent of $s$, while the radial mode $r-1$ is increasingly penalized as $s$
increases.

This explicitly illustrates the stability results of
Section~\ref{subsec:change-sobolev}.

%-------------------------------------------------
\subsection{Concluding remarks}

These explicit computations illustrate:
\begin{itemize}
\item the collapse of transverse degrees of freedom in the thin limit;
\item the block-diagonal Fourier structure induced by rotational invariance;
\item the stability of Galerkin schemes under changes of geometry, Sobolev order,
and thinness parameter.
\end{itemize}
They provide a concrete complement to the abstract results of the main text.
\vskip 12pt

\paragraph{\bf Fundings:} No funding is available.

\vskip 12pt

\paragraph{\bf Data availability statement.} No data is available for this work.

\vskip 12pt

\paragraph{\bf Conflict of interest statement.} The author declares no conflict of interest.

\vskip 12pt

\paragraph{\bf Acknowledgements:} J.-P.M thanks the France 2030 framework programme Centre Henri Lebesgue ANR-11-LABX-0020-01 
for creating an attractive mathematical environment.

\vskip 12pt

\paragraph{\bf Author's Note on AI Assistance.}
Portions of the text were developed with the assistance of a generative language model (OpenAI ChatGPT). The AI was used to assist with drafting, editing, and standardizing the bibliography format. All mathematical content, structure, and theoretical constructions were provided, verified, and curated by the author. The author assumes full responsibility for the correctness, originality, and scholarly integrity of the final manuscript.

\end{document}